\documentclass[conference]{IEEEtran}

\usepackage[cmex10]{amsmath}
\usepackage[dvipdfmx]{graphicx}
\usepackage{array}
\usepackage{mdwmath}
\usepackage{mdwtab}
\usepackage{amssymb}

\newtheorem{theorem}{Theorem}
\newtheorem{lemma}[theorem]{Lemma}
\newtheorem{definition}[theorem]{Definition}
\newtheorem{corollary}[theorem]{Corollary}
\newtheorem{remark}[theorem]{Remark}
\newtheorem{proposition}[theorem]{Proposition}
\newcommand{\qed}{\hfill$\square$}

\newenvironment{proof}{%
 \noindent{\em Proof.\ }}{%
 \hspace*{\fill}\qed \\
 \vspace{2ex}}
\DeclareMathAlphabet{\bm}{OML}{cmm}{b}{it}

\newcommand{\bol}[1]{\mathbf{#1}}

\begin{document}
%
\title{Expurgation Exponent of Leaked Information in Privacy Amplification for Binary Sources}



%
\author{\IEEEauthorblockN{Shun Watanabe\IEEEauthorrefmark{1}}
\IEEEauthorblockA{\IEEEauthorrefmark{1}
Department of Information Science and Intelligent Systems, 
University of Tokushima, Tokushima, Japan, \\
Email: shun-wata@is.tokushima-u.ac.jp }}


\maketitle

\begin{abstract}
We investigate the privacy amplification problem in which Eve can
observe the uniform binary source through a binary erasure channel (BEC)
or a binary symmetric channel (BSC). For this problem, we derive the
so-called expurgation exponent of the information leaked to Eve.
The exponent is derived by relating the leaked information
to the error probability of the linear code that is generated by the linear
hash function used in the privacy amplification, which is also interesting in its own right.
 The derived exponent is larger than state-of-the-art exponent recently derived by Hayashi at low rate.
\end{abstract}


%
\IEEEpeerreviewmaketitle

\section{Introduction}

In information theoretic key 
agreement problem \cite{maurer:93,ahlswede:93,csiszar:00,csiszar:04,csiszar:08,gohari:10}, 
legitimate parties need to distill a secret key from
a random variable in the situation such that
an eavesdropper can access to
a random variable that is correlated to the legitimate parties' random 
variable. The privacy amplification is a technique to distill a secret
key under the situation by using a (possibly random) 
function \cite{bennett:95}.
The security of distilled key is evaluated by various kinds of measures.
In this paper, we focus on the leaked information,
which is the mutual information between the distilled key and eavesdropper's  
random variable (the so-called strong security \cite{maurer:94b,csiszar:96}),  
because it is the strongest notion among security 
criterion \cite{csiszar:04} (see also \cite[Appendix 3]{hayashi:10}).

The privacy amplification is usually conducted by using a family
of universal 2 hash functions \cite{carter:79}.
In \cite{bennett:95}, Bennett {\em et.~al.} evaluated 
ensemble averages of the leaked information for universal 2 families, 
and derived an upper bound on the leaked information by using the 
R\'enyi entropy of order $2$. In \cite{renner:05d}, Renner and Wolf
evaluated ensemble averages of the leaked information for universal 2 families, 
and derived an upper bound on the leaked information by using the 
smooth minimum entropy. In \cite{hayashi:10}, Hayashi evaluated 
ensemble averages of the leaked information for universal 2 families, 
and derived a parametric upper bound on the leaked information by using the 
R\'enyi entropy of order $1+\theta$.  Concerning the exponential decreasing
rate of the leaked information, the exponent derived by Hayashi's bound is 
state-of-the-art.

In noisy channel coding problem, the exponential decreasing rate
of the error probability is also regarded as an important performance criterion of codes, 
and has been studied for a long time.  The best exponent at high rates
is the one derived by Gallager's random coding
bound \cite{gallager:65}. However, Gallager's exponent is not tight in general,
and can be improved at low rates because the random code ensemble 
involves some bad codes and those bad codes become dominant at low rates.
The improved exponent by expurgating those bad codes is usually called 
the expurgation exponent \cite{gallager:65,barg:02}. 
Similar improved exponents are also known in the context of
the Slepian-Wolf coding \cite{csiszar:81,csiszar:82} or the quantum 
error correction \cite{barg:02b}.

The purpose of this paper is to show a security analog of above results, i.e.,
to derive an improved exponent
of the leaked information in the privacy amplification at low rates.
For this purpose, we concentrate our attention on the case such 
that the random variable possessed by the legitimate parties is
the binary uniform source and the function used in the privacy amplification
is a linear matrix.

We first consider the case such that
the eavesdropper's random variable is generated via a binary erasure
channel (BEC). For this case, we first relate the leaked information to the maximum likelihood (ML)
decoding error probability of the linear code whose generator matrix is the one used in
the privacy amplification. Then an improved exponent is derived by using the result
of the expurgation exponent of linear codes. 

It should be noted that a similar approach to relate the leaked information
to the erasure error correction has been appeared in \cite{thangaraj:07}. However in this paper,
we directly relate the leaked information to the ML decoding error probability, which enables
us to derive the improved exponent. It should be also noted that the approach in this
paper is completely different from the error correction approach conventionally 
used to prove the so-called weak security and the problem pointed out in \cite{watanabe:10}
does not apply to our approach.

Next, we consider the case such that
the eavesdropper's random variable is generated via a binary symmetric channel (BSC).
For this case, the technique used in the BEC case cannot be directly applied.
Thus, we first reduce the BSC case to the BEC case by using the partial order 
between BSCs and BECs. The reduction turns out to be quite tight. Indeed, 
the exponent derived via this reduction is as good as Hayashi's exponent below
the critical rate, and strictly better than Hayashi's exponent below the expurgation rate,
which resemble the relation between the expurgation exponent  and the random coding
exponent of the noisy channel coding.
Our results suggest that the privacy amplification with a universal 2 family
is not necessarily optimal.

The rest of the paper is organized as follows.
We first explain the problem formulation of the privacy amplification
in Section \ref{section:problem}.
Then, we consider the BEC case and the BSC case in 
Sections \ref{section:erasure} and \ref{section:bsc} respectively.
Conclusions are discussed in Section \ref{section:conclusion}.

\section{Problem Formulation}
\label{section:problem}

Let $(X^n,Z^n)$ be a correlated i.i.d. source with distribution $P_{XZ}$.
The alphabet is denoted by ${\cal X} \times {\cal Z}$.
In the privacy amplification problem, we are interested in generating 
the uniform random number on ${\cal S}_n$ by using a 
function $f_n : {\cal X}^n \to {\cal S}_n$. The joint distribution of the generated
random number and the side-information is given by
\begin{eqnarray*}
P_{S_n Z^n}(s_n,z^n) = \sum_{x^n \in f_n^{-1}(s_n)} P_{XZ}^n(x^n,z^n),
\end{eqnarray*}
where $f_n^{-1}(s_n) = \{x^n \in {\cal X}^n : f_n(x^n) = s_n \}$.

The security is evaluated by the leaked information
\begin{eqnarray*}
I(f_n) = I(S_n; Z^n)
\end{eqnarray*}
where $I(\cdot; \cdot)$ is the mutual information and
we take the base of the logarithm to be $e$.

For given rate $R \ge 0$,
we are interested in the exponential decreasing rate
of $I(f_n)$, i.e.,
\begin{eqnarray*}
\lefteqn{ E( R; X|Z)  } \\
&=& \sup\left\{ \liminf_{n \to \infty} - \frac{1}{n}  \log I(f_n) : \liminf_{n \to \infty} \frac{1}{n} \log |{\cal S}_n| \ge R  \right\}.
\end{eqnarray*}

In the privacy amplification problem, we typically use the universal $2$ hash family.
\begin{definition}
A family ${\cal F}_n$ of functions $f_n:{\cal X}^n \to {\cal S}_n$ is called universal $2$
if
\begin{eqnarray*}
\Pr\{ F_n(x^n) = F_n(\hat{x}^n) \} \le \frac{1}{|{\cal S}_n|}
\end{eqnarray*}
for every $x^n \neq \hat{x}^n$, where $F_n$ is the uniform random
variable on ${\cal F}_n$.
\end{definition}

For parameter $\theta$, let 
\begin{eqnarray*}
\psi(\theta; X|Z) &=& - \log \sum_{x,z} P_{ZX}(x,z)^{1+\theta} P_Z(z)^{-\theta} \\
	&=& - \log \sum_{x,z} P_{XZ}(x,z) \exp\left[ \theta \log P_{X|Z}(x|z) \right]. 
\end{eqnarray*}

Hayashi derived the following lower bound on $E( R; X|Z)$. 
\begin{proposition}[\cite{hayashi:10}]
\label{proposition:hayashi-exponent}
For any universal $2$ hash family ${\cal F}_n$, 
we have
\begin{eqnarray*}
E(R;X|Z) &\ge& \liminf_{n \to \infty} - \frac{1}{n} \log  \mathbb{E}_{{\cal F}_n}[I(f_n)] \\
	&\ge& E_r(R;X|Z) \\ 
	&:=& \max_{0 \le \theta \le 1} \left[ \psi(\theta; X|Z) - \theta R \right],
\end{eqnarray*}
where $\mathbb{E}_{{\cal F}_n}$ means the average over randomly chosen 
function from ${\cal F}_n$.
\end{proposition}

\section{Side-Information via Binary Erasure Channel}
\label{section:erasure}

In this section, we assume that $X$ is the uniform binary source
and $Z$ is the output of the binary erasure channel (BEC) with erasure probability $\varepsilon$, i.e.,
$P_{XZ}(x,x) = \frac{1- \varepsilon}{2}$ and $P_{XZ}(x, ?) = \frac{\varepsilon}{2}$,
where $?$ represent the erasure symbol (see Fig.~\ref{Fig:bec-e}).
For given sequence $z^n$, let ${\cal J}(z^n) \subset \{1,\ldots, n \}$ be
the set of those indices such that $z_j = ?$. When the sequence $z^n$
is obvious from the context, we abbreviate ${\cal J}(z^n)$ as ${\cal J}$.
\begin{figure}[tbp]
\begin{tabular}{lr}
 \begin{minipage}{0.45\hsize}
  \begin{center}
   \includegraphics[width=0.9\linewidth]{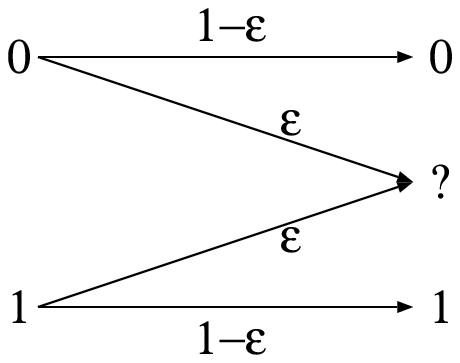}
  \end{center}
  \caption{The channel considered in Section \ref{section:erasure}.}
  \label{Fig:bec-e}
 \end{minipage} &
  \begin{minipage}{0.45\hsize}
  \begin{center}
   \includegraphics[width=0.9\linewidth]{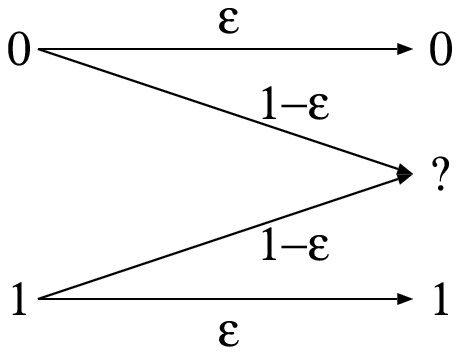}
  \end{center}
  \caption{The virtual channel considered in Section \ref{section:erasure}.}
  \label{Fig:bec-1-e}
 \end{minipage}
 \end{tabular}
\end{figure}

In the rest of this paper, we concentrate on the linear function $f_n:{\cal X}^n \to {\cal S}_n$.
Thus, we implicitly assume that ${\cal X} = \mathbf{F}_2$ and ${\cal S}_n = \mathbf{F}_2^k$
for some $k$, where $\mathbf{F}_2$ is the field of order $2$.
Let $M_n$ be $k \times n$ matrix with entries in $\mathbf{F}_2$.
We consider function $f_n : x^n \to x^n M^T_n$ and the security 
criterion is denoted by $I(M_n)$.
The sequence $x^n_{{\cal J}}$ is a subsequence of $x^n$ that consist of
the indices in ${\cal J}$, and the matrix $M_{\cal J}$ is a sub-matrix of $M_n$
that consist of the columns in ${\cal J}$.   

The following lemma was presented by Ozarow and Wyner.
\begin{lemma}[\cite{ozarow:84}]
\label{lemma:ozarow}
We have
\begin{eqnarray*}
H(S_n| Z^n = z^n) \ge \mbox{rank}(M_{{\cal J}(z^n)})
\end{eqnarray*}
for every $z^n$.
\end{lemma}

We consider the virtual BEC
with erasure probability $1 - \varepsilon$
(see Fig.~\ref{Fig:bec-1-e}), i.e., 
$P_{Y|X}(x|x) = \varepsilon$ and $P_{Y|X}(?|x) = 1- \varepsilon$.
From Lemma \ref{lemma:ozarow}, we have the following.
\begin{theorem}
\label{theorem:duality}
Let ${\cal C}_n$ be the linear code whose generator matrix is $M_n$, and
let $P_{ML}({\cal C}_n,1 - \varepsilon)$ be the maximum likelihood decoding error probability\footnote{Ties are
counted as errors.} 
of the code ${\cal C}_n$ over the BEC($1 - \varepsilon$). Then, we have
\begin{eqnarray*}
I(M_n) \le n P_{ML}({\cal C}_n, 1 - \varepsilon).
\end{eqnarray*}
\end{theorem}
\begin{proof}
Let $m^k \in \mathbf{F}_2^k$ is a message to be sent,
and the encoded message $m^k M_n$ is sent over the BEC($1 - \varepsilon$).
Suppose that the received signal is $y^n$.
If $\mbox{rank}(M_{{\cal J}(y^n)^c}) = k$, then the ML decoder
output $m^k$, where
${\cal J}(y^n)^c = \{1,\ldots,n\} \backslash {\cal J}(y^n)$ is the non erased bits.
On the other hand, if $\mbox{rank}(M_{{\cal J}(y^n)^c}) < k$,
there are plural messages that are compatible with $y^n$, and thus the
ML decoder fail to output $m^k$. Therefore, the ML decoding error probability
can be written as 
\begin{eqnarray*}
\lefteqn{ P_{ML}({\cal C}_n, 1 - \varepsilon) } \\
&=& \sum_{{\cal J}^c \subset \{1,\ldots,n \}} 
	 (1 - \varepsilon)^{n - |{\cal J}^c|} \varepsilon^{|{\cal J}^c|} \bol{1}[\mbox{rank}(M_{{\cal J}^c}) < k].
\end{eqnarray*}
On the other hand , by using Lemma \ref{lemma:ozarow} and
by noting that $H(S_n) \le n$, we have
\begin{eqnarray*}
I(M_n) \le n \sum_{{\cal J} \subset \{1,\ldots,n\} } (1 - \varepsilon)^{n- |{\cal J}|} \varepsilon^{|{\cal J}|}
	\bol{1}[\mbox{rank}(M_{{\cal J}}) < k].
\end{eqnarray*}
Thus, we have the assertion of the theorem.
\end{proof}

By using a linear code achieving the Gilbert-Varshamov bound, we
have the following.
\begin{corollary}
\label{corollary:expurgation-exponent}
There exists a linear function $f_n: x^n \to x^n M_n^T$ such that
\begin{eqnarray}
\lefteqn{ E( R;X|Z) } \nonumber \\
	&\ge& \liminf_{n \to \infty} - \frac{1}{n} \log I(f_n) \\
	&\ge& \liminf_{n \to \infty} - \frac{1}{n} \log P_{ML}({\cal C}_n,1 - \varepsilon) \\
	&\ge& E_x(R,1 - \varepsilon)  
	\label{eq:expurgation-exponent-erasure} \\
	&:=& \max_{\theta \ge 1}\left[ \theta\{ \log 2 - R - \log (1 + (1-\varepsilon)^{1/\theta}) \} \right].
\end{eqnarray}
\end{corollary}
\begin{proof}
First note that the error probability of the channel coding and
that of Slepian-Wolf coding (with full side-information) are the same for
linear code and BEC. Thus, Csisz\'ar's linear Slepian-Wolf code result \cite{csiszar:82}
implies that there exists a code satisfying
\begin{eqnarray}
\lefteqn{ \liminf_{n \to \infty} - \frac{1}{n} \log P_{ML}({\cal C}_n,1 - \varepsilon) } \nonumber \\
	&\ge& \min_{H(W) \ge \log 2 - R} \left[ (\log 2- R) - H(W) + \phantom{\sum_{x,y}} \right. \nonumber \\
	&& \left. \mathbb{E}\left[ - \log \sum_{x,y} \sqrt{P_{XY}(x,y)P_{XY}(x+W,y)}\right]  \right] \nonumber \\
	&=& \min_{h( p) \ge \log 2 - R} \left[ - p \log (1-\varepsilon) + (\log 2 -R) - h( p) \right], \nonumber \\
	\label{eq:ex-bound-min-representation}
\end{eqnarray} 
where we set $P_W(1) = p$. Since the objective function of Eq.~(\ref{eq:ex-bound-min-representation})
is convex, by introducing 
\begin{eqnarray*}
L(\lambda) := \min_p \left[ -p \log (1-\varepsilon) + (1+\lambda) (\log 2 - R - h( p)) \right]
\end{eqnarray*}
for $\lambda \ge 0$, Eq.~(\ref{eq:ex-bound-min-representation}) can be written \cite{boyd-book:04} as
\begin{eqnarray*}
\max_{\lambda \ge 0} L(\lambda).
\end{eqnarray*}
By changing the variable as $\theta = 1 + \lambda$, Eq.~(\ref{eq:ex-bound-min-representation}) can be 
also written as 
\begin{eqnarray*}
\max_{\theta \ge 1} L(\theta - 1) =
\max_{\theta \ge 1}\left[ \theta\{ \log 2 - R - \log (1 + (1-\varepsilon)^{1/\theta}) \} \right].
\end{eqnarray*}
\end{proof}
Note that $E_x(R, 1 - \varepsilon)$ is the expurgation exponent
for BEC($1-\varepsilon$) \cite{gallager:68}.

\begin{remark}
It should be noted that
\begin{eqnarray}
\lefteqn{ E_r( R; X|Z) } \nonumber \\
	&=& E_r(R, 1 - \varepsilon) 
	\label{eq:random-coding-exponent-erasure} \\
	&:=& \max_{0 \le \theta \le 1}\left[ - \log\left\{ (1- \varepsilon) + \frac{1}{2^\theta} \varepsilon \right\} - \theta R \right].
\end{eqnarray}
Since $E_r(R, 1 - \varepsilon)$ is the random coding exponent for BEC($1-\varepsilon$) \cite{gallager:68}, 
Hayashi's exponent can be also derived from Theorem \ref{theorem:duality}.
\end{remark}

From Eq.~(\ref{eq:expurgation-exponent-erasure}) and Eq.~(\ref{eq:random-coding-exponent-erasure})
and known facts on the exponents, we find that the exponent of PA in Corollary \ref{corollary:expurgation-exponent}
is larger than that in Proposition \ref{proposition:hayashi-exponent} for low $R$.
These exponents are compared in Fig.~\ref{Fig:EX-ER-05} for $\varepsilon = 0.5$.
We find that $E_x(R, 1 - \varepsilon)$ is strictly larger than $E_r(R, 1 - \varepsilon)$
at low rates.
\begin{figure}[t]
\centering
\includegraphics[width=\linewidth]{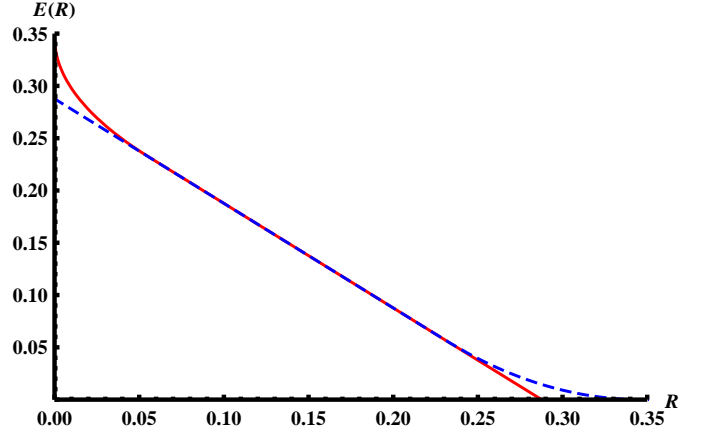}
\caption{Comparison of $E_r(R, 1 - \varepsilon)$ (dashed line) and $E_x(R,1- \varepsilon)$
(solid line) for $\varepsilon = 0.5$.
}
\label{Fig:EX-ER-05}
\end{figure}

\section{Side-Information via Binary Symmetric Channel}
\label{section:bsc}

\begin{figure}[tbp]
\begin{tabular}{lr}
 \begin{minipage}{0.45\hsize}
  \begin{center}
   \includegraphics[width=0.9\linewidth]{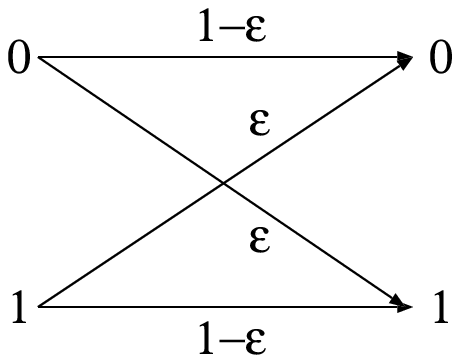}
  \end{center}
  \caption{The channel considered in Section \ref{section:bsc}.}
  \label{Fig:bsc-e}
 \end{minipage} &
  \begin{minipage}{0.45\hsize}
  \begin{center}
   \includegraphics[width=0.9\linewidth]{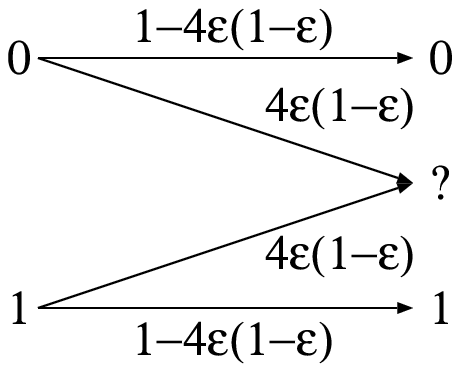}
  \end{center}
  \caption{The virtual channel considered in Section \ref{section:bsc}. This channel is less noisy than the BSC in Fig.~\ref{Fig:bsc-e}.}
  \label{Fig:bec-convolution}
 \end{minipage}
 \end{tabular}
\end{figure}

In this section, we assume that $X$ is the uniform binary source
and $Z$ is the output of the binary symmetric channel (BSC) with crossover probability $\varepsilon$, i.e.,
$P_{XZ}(x,x) = \frac{1- \varepsilon}{2}$ and $P_{XZ}(x, x+1) = \frac{\varepsilon}{2}$ (see Fig.~\ref{Fig:bsc-e}).
Let $\bar{Z}$ be the output of BEC($4 \varepsilon (1- \varepsilon)$) with input $X$.
Since BEC($4 \varepsilon (1- \varepsilon)$) (see Fig.~\ref{Fig:bec-convolution}) is less noisy than 
BSC($\varepsilon$) \cite{nair:10},
we have
\begin{eqnarray*}
I(S_n; Z^n) 
 \le I(S_n; \bar{Z}^n). 
\end{eqnarray*}
Thus, Corollary \ref{corollary:expurgation-exponent} can be applied to
the case considered in this section.
\begin{theorem}
\label{theorem:expurgation-bsc-less-noisy-reduction}
Let $\bar{Z}$ be the output of BEC($4 \varepsilon (1 - \varepsilon)$) with input $X$.
Then, we have
\begin{eqnarray*}
E( R; X|Z) &\ge& E( R; X|\bar{Z}) \\
	&\ge& E_x(R, 1 - 4 \varepsilon(1 - \varepsilon)).
\end{eqnarray*}
\end{theorem}

Hayashi's exponent for BSC($\varepsilon$) is
\begin{eqnarray*}
E_r( R; X|Z) = \max_{0 \le \theta \le 1}\left[ - \log \left\{ (1-\varepsilon)^{1+\theta} + \varepsilon^{1+\theta} \right\} - \theta R \right].
\end{eqnarray*}
The exponents are compared in Fig.~\ref{Fig:EX-ED-Less-Noisy-Reduction-011} 
and Fig.~\ref{Fig:EX-ED-Less-Noisy-Reduction-025} for
$\varepsilon = 0.11$ and $0.25$ respectively. 
\begin{figure}[t]
\centering
\includegraphics[width=\linewidth]{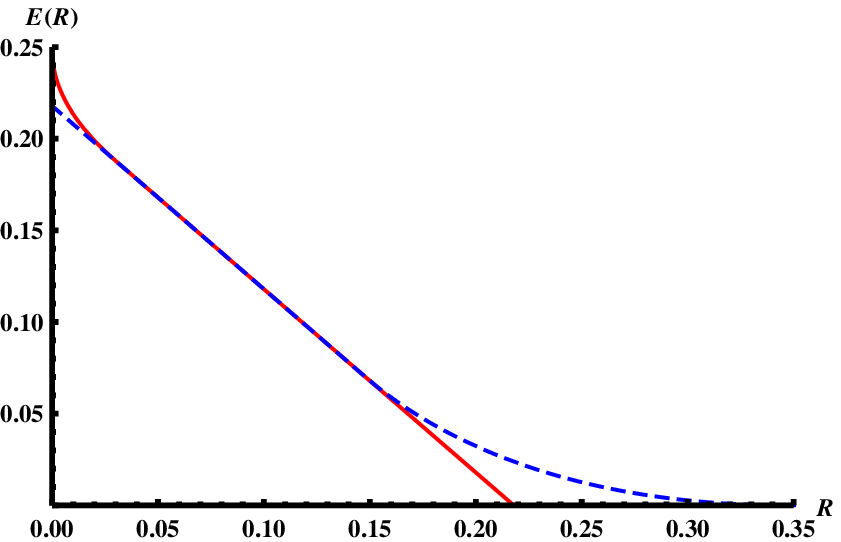}
\caption{Comparison of $E_r(R, X|Z)$ (dashed line) and $E_x(R,1- 4 \varepsilon(1 - \varepsilon))$
(solid line) for BSC($0.11$). 
}
\label{Fig:EX-ED-Less-Noisy-Reduction-011}
\end{figure}
\begin{figure}[t]
\centering
\includegraphics[width=\linewidth]{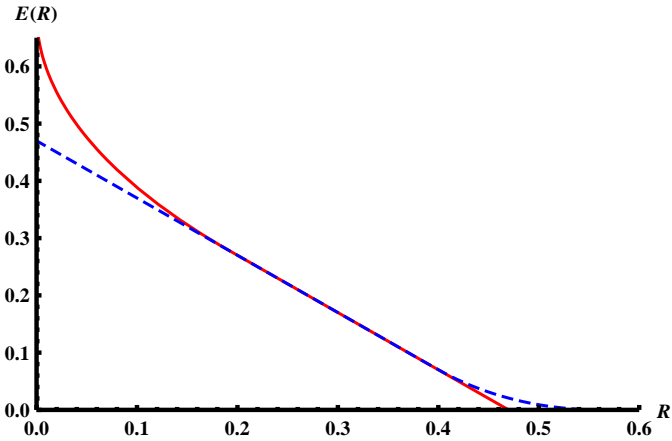}
\caption{Comparison of $E_r(R, X|Z)$ (dashed line) and $E_x(R,1- 4 \varepsilon(1 - \varepsilon))$
(solid line) for BSC($0.25$).
}
\label{Fig:EX-ED-Less-Noisy-Reduction-025}
\end{figure}

Let $R_{cr}(\varepsilon)$ be the critical rate, i.e.,
the largest rate such that 
the optimization in $E_r(R;X|Z)$ is achieved by $\theta = 1$.
Then, for $R \le R_{cr}(\varepsilon)$, we have
\begin{eqnarray*}
E_r(R;X|Z) = - \log\{ (1-\varepsilon)^2 + \varepsilon^2 \} - R.
\end{eqnarray*}
On the other hand, let $R_x(\varepsilon)$ be the expurgation rate, i.e., the smallest rate
such that the optimization in $E_x(R, 1 - 4 \varepsilon(1- \varepsilon))$
is achieved by $\theta = 1$. Then, for
$R_x(\varepsilon) \le R$, we have
\begin{eqnarray*}
\lefteqn{ E_x(R,1 - 4 \varepsilon (1 - \varepsilon)) } \\
	&=& \log 2 - R - \log\left( 1 + 1 - 4 \varepsilon (1 - \varepsilon) \right) \\
	&=& - \log \{ (1- \varepsilon)^2  + \varepsilon^2) \} - R.
\end{eqnarray*}
Thus, for $R_x(\varepsilon) \le R \le R_{cr}(\varepsilon)$, 
$E_r(R;X|Z) = E_x(R, 1 - 4 \varepsilon (1 - \varepsilon))$, which
can be also observed in Fig.~\ref{Fig:EX-ED-Less-Noisy-Reduction-011} 
and Fig.~\ref{Fig:EX-ED-Less-Noisy-Reduction-025}.
We also find that $E_x(R, 1 - 4 \varepsilon (1- \varepsilon))$ is strictly
larger than $E_r(R; X|Z)$ at low rates.

\section{Conclusion}
\label{section:conclusion}

For the BEC case and the BSC case, 
we derived the expurgation exponent of the leaked 
information in the privacy amplification. 
The technique to relate the leaked information
to the ML decoding error probability heavily relies on 
the specific structure of the BEC. Thus, to derive the 
expurgation exponent for general cases, 
a method to expurgate bad functions directly 
might be needed.

Hayashi derived a quantum counter part of 
Proposition \ref{proposition:hayashi-exponent} in \cite{hayashi:12b}.
It is also interesting to derive the expurgation exponent in the privacy 
amplification for quantum adversary. 
For the case such that the eavesdropper's information is
generated via the complementary channel of a Pauli channel, the technique to relate the leaked
information to the ML decoding error probability is already 
known \cite{hayashi:06}\footnote{Although a result in the classical information theory
is typically a commutative special case of the quantum counter part, this is not the case
for the result shown in this paper. Indeed, the result in \cite{hayashi:06} is derived by using
the relation between the eavesdropper's information gain and the amount of phase error 
caused in the main channel, which is a unique feature of quantum mechanics and there is
no classical counter part.},
and it is not difficult to derive the expurgation exponent.
In general, more refined technique is needed.
These topics will be investigated in elsewhere.

\section*{Acknowledgment}

The author would like to thank Prof.~Yasutada Oohama
for valuable discussion. 
The author also would like to thank Prof.~Prakash Narayan
for inviting the author to the workshop.
This research is partly supported by 
Grand-in-Aid for Young Scientist(B):2376033700 and
Grand-in-Aid for Scientific Research(A):2324607101.


\end{document}